\documentclass[11pt]{article}

\usepackage{amsmath,amssymb,amsthm,boxedminipage,enumerate}

\usepackage{fullpage}

\usepackage{algorithm}

\usepackage{algorithmic}

\usepackage{xspace}

\usepackage{amsfonts}

\usepackage{graphicx}

\newtheorem{theorem}{Theorem}[section]

\newtheorem{lemma}[theorem]{Lemma}

\newtheorem{corollary}[theorem]{Corollary}

\theoremstyle{definition}

\newtheorem{observation}{Observation}

\newcommand{\Search}{\mathsf{Search}}

\newcommand{\Insertedge}{\mathsf{Insert\_edge}}

\newcommand{\Deleteedge}{\mathsf{Delete\_edge}}

\newcommand{\Refresh}{\mathsf{Refresh}}

\newcommand{\old}[1]{}

\title{Dynamic graph connectivity with improved worst case update time and sublinear space}

\author{David Gibb
\thanks{Department of Computer Science, University of Victoria, BC, Canada; david90@uvic.ca, bmkapron@uvic.ca, val@uvic.ca, nolandthorn@gmail.com. This research was supported by NSERC and a Google Faculty Research Grant}
\and Bruce Kapron
 \and Valerie King\ $^{*}$ \and Nolan Thorn $^{*}$}

\date{}
\begin{document}

\maketitle

\thispagestyle{empty}
\begin{abstract} 
\small\baselineskip=9pt

This paper considers fully dynamic graph algorithms with both faster worst case update time and sublinear space. The fully dynamic graph connectivity problem is the following: given a graph on a fixed set of $n$ nodes, process an online sequence of edge insertions, edge deletions, and queries of the form ``Is there a path between nodes $a$ and $b$?" In 2013, the first data structure was presented with worst case time per operation which was polylogarithmic in $n$. In this paper, we shave off a factor of
$\log n$ from that time, to $O(log^4 n)$ per update. For sequences which are polynomial in length, our algorithm answers queries in O($\log n/\log\log n$) time correctly with high probability and using $O(n \log^2n)$ words (of size $\log n$). This matches the amount of space used by the most space-efficient graph connectivity streaming algorithm. We also show that 2-edge connectivity can be maintained using $O(n \log^2 n)$ words with an amortized update time of $O(\log^6 n)$. 

\end{abstract}
\newpage
\setcounter{page}{1}
\section{Introduction}

In the dynamic connectivity problem, we are given an undirected graph $G=(V,E)$, where $V$ is a fixed set of  $n$ nodes, and an online sequence of updates and queries of the following form:

\begin{itemize}
\setlength\itemsep{0em}
\item
$\mathsf{Delete}(e)$: Delete edge $e$ from $E$.
\item
$\mathsf{Insert}(e)$: Insert edge $e$ into $E$.
\item
$\mathsf{Query}(x,y)$: Is there a path between nodes $x$ and $y$? 
\end{itemize}
The goal is to process graph updates and queries as efficiently as possible. 
Our data structure solves the dynamic connectivity problem using $O(\log^4)$ worst case time per update operation and $O(\log  n/\log \log n)$ worst case time per query.  In addition, if the sequence of updates is polynomial in length, our data structure needs only a sublinear amount of storage, namely $O(n \log^2 n) $ words (of size $O(\log n)$), to ensure that with high probability, all queries are answered correctly. By way of comparison,  the dynamic graph stream problem is to process an online sequence of polynomial number  of edge insertions and deletions so that when the stream ends, a spanning forest is output, where the goal is to minimize space. Our algorithm is the first sublinear space dynamic graph algorithm: it matches the space needed by the most space efficient dynamic graph stream algorithm by Ahn, Guha, and McGregor from 2012 \cite{AGM}, and it shaves a factor of $O(\log n)$ off the worst case time per update of the best previously known dynamic connectivity algorithm by Kapron, King, and Mountjoy from 2013 \cite{KKM}. 

 As an application, we show the first sublinear space dynamic  2-edge connectivity algorithm, which uses $O(\log^6 n)$ amortized update time, $O(\log n /\log\log n)$ query time,  and $O(n \log^2 n)$ words. This compares to the $O(m+ n \log n)$ word algorithm with amortized update time $O(\log^4 n)$ and $O(\log n/\log\log n)$ query time by Holm, de Lichtenberg, and Thorup from 1998 \cite{Holm}.

The sublinear space algorithm has 2-sided error of $1/n^c$ for any constant $c$. If we keep a list of edges with additional space $O(m)$
then there is 1-sided error:
if the answer to a query is ``yes," then the answer is always correct; if it is ``no," then it is correct with probability $1-1/n^c$ for $c$ any constant.  In both cases, we assume the adversary knows the edges in the graph, but constructs the sequence of updates and queries with no knowledge of the random bits used by the algorithm, other than whether the queries have been answered correctly or not. The edges of the spanning forest used by the data structure are not revealed to the adversary.

Several dynamic graph problems can easily be reduced to the problem of maintaining connectivity.  See \cite{HK}. Thus our algorithm yields a dynamic graph algorithm for determining the weight of a minimum spanning tree in a graph with $k$ different edge weights, where the cost of updates is increased by a factor of $k$. It also can be used to maintain the property of bipartiteness with only constant factor increase in cost. See \cite{AGM},\cite{HK}.   

 \bigskip

\noindent
{\it An Overview of Techniques:} 
Nearly every dynamic connectivity data structure maintains a spanning forest $F$. Dealing with edge insertions is relatively easy. The challenge is to find a replacement edge when a tree edge is deleted, splitting a tree into two subtrees.  A replacement edge is an edge reconnecting the two subtrees, or, in other words, in the {\em cutset} of the cut $(T, V\setminus T)$ where $T$ is one of the subtrees. To simplify notation,
where there is no confusion, we allow $T$ to denote the set of nodes in the tree $T$. Otherwise we will use the notation $v(T)$.


We compare the techniques of our paper to the 2013 dynamic graph connectivity \cite{KKM}
and the 2012 dynamic graph stream paper \cite{AGM}. 
All three papers use random sampling to identify the name of a edge which crosses a cut.  The  2013 dynamic graph paper introduced the cutset data structure. It maintains a dynamic forest of disjoint trees which is a subgraph of a graph $G$
and, for each tree, a ``leaving" edge, i.e., an edge with exactly one endpoint in the tree. 
That paper uses ${O}(m \log n + n \log^2 n )$ words to implement the cutset data structure and ensures that edges leaving are identified with high probability. For each edge inserted, it flips fresh coins to determine which samples it appears in and records these outcomes in a table. 

The graph stream paper implicitly uses a restricted cutset data structure which accommodates a continuous stream of nontree edge insertions and deletions, with one final batch change of nontree edges to tree edges, at which point the leaving edge from a tree is determined with constant probability. The streaming paper uses $O(n \log n)$ words to implement this restrictive cutset data structure and $O(n \log^2 n)$ words overall. It uses hash functions rather than independent random bits to determine whether an edge is in a sample, eliminating the need for a table to remember the sampling outcomes. It uses a well-studied technique known as $L_0$ sampling and 1-sparse recovery to sample, find an edge with constant probability, and then verify  with high probability it is an edge. 

Our cutset data structure also determines an edge leaving with constant probability
using $O(n \log n)$ words. It can be seen as using $L_0$ sampling and a novel, possibly more practical, method of 1-sparse recovery.

All three papers use  $O(\log n)$ tiers of cutset data structures to build a spanning forest, similar to the Bor\r{u}vka's parallel minimum spanning tree algorithm. Each tier contains a forest, and each forest is a subforest of the tier above it.  In the cutset data structure on tier $i$, the leaving edges discovered there become tree edges in the forest on tier $i+1$, and so on. The streaming paper needs to form these forests only once at the end of the stream. In worst case, a single insertion or deletion can cause a superpolylogarithmic number of changes to edges leaving on the various tiers. Thus maintaining the spanning forest for each update that the streaming algorithm would have produced if the stream had ended at that update, is too costly in terms of time. By using high probability to find edges leaving, the 2013 dynamic graph paper avoids this problem by finding an edge leaving every non-maximal tree in a tier at an extra cost of $\log n$ factor in time and space. Our algorithm avoids the cost of using high probability; it does not necessarily keep the same tree edges as the streaming algorithm, but it ensures that about the same number of trees are merged on the next tier.  This induces a subtle change to the invariants, which are less restrictive.

\bigskip

\noindent
{\it Related work:}
For other related work on dynamic connectivity, see \cite{KKM}

\bigskip

\noindent
{\it Organization:} In Section 2 we define  the cutset data structure. In Section 3, we show how it can be used to solve dynamic connectivity. In Section 4, we discuss the implementation of cutset data structures. In Section 5, we give a space efficient dynamic 2-edge connectivity algorithm.  

\section{Cutset data structure definition and results}

Consider a forest $F$ of disjoint trees which are subtrees of a graph $G$. $G$ and $F$ are dynamic, that is, they are undergoing updates in the form of edge insertions and deletions.  In addition, an edge in $G$ which links two trees may be updated to become a tree edge which would then cause the two trees in $F$ to be joined as one.  Similarly a tree edge may be changed to a non-tree edge. 
The problem is to design a data structure with the goal of maintaining, for each tree $T$ in $F$ which is non-maximal, an edge in the cutset of $(T, V\setminus T)$. A {\it successful} $\Search(T)$ returns this edge if the cutset is non-empty and $null$  otherwise. We assume that the graph updates are oblivious to the (randomized) outcomes of $\Search$. I.e., we may consider the sequence of updates to be fixed in advance, but updates are revealed to the algorithm one at a time.  The operations are:

\begin{itemize}
\setlength\itemsep{0em}
\item 
$\mathsf{Insert\_edge}(x,y )$: Adds edge $\{x,y\}$ to $E$.

\item
$\mathsf{Make\_tree\_edge}(x,y)$:  For $\{x,y\} \in E$ but not in $F$, adds $\{x,y\}$ to $F$.

\item
$\mathsf{Make\_nontree\_ edge}(x,y)$: Removes $\{x,y\}$ from $F$.

\item
$\mathsf{Delete\_ edge}(x,y)$: Removes $\{x,y\}$ from $E$.  If $\{x,y\}$ is an edge of $F$, removes it from $F$.

\item
$\mathsf{Search}(T)$: For $T \in F$, returns $null$ if there are no edges of $E$ in the cutset of $(T, V\setminus T)$; returns an edge in th cutset of $(T, V\setminus T)$ otherwise. 

\end{itemize}

\begin{lemma} \label{l:cutset}
\label{fixedcut}
Let $c>0$ be any constant.  There is a cutset data structure for graph $G$ with forest $F$ with worst case update time $O(\log^2 n)$ and space $O(n\log n)$ words, such that for any $T \in F$, where $C$ is the cutset of $(T,V\setminus T)$, if
 $C\neq \emptyset$, then with constant probability $p\geq 1/8 -1/n^c$, $\Search(T)$ returns an edge in the cutset; with probability less than $1/n^c$ it returns a non-edge or edge not in $C$; otherwise it returns $null$. If $C=\emptyset$, it always returns $null$.
   
With an auxilliary list of $O(m)$ words to keep track of the current edge set in the graph, $\Search(T)$ returns an edge in $C$ with probability $1/8$ and always returns an edge in $C$, if it returns an edge. 
\end{lemma}

\noindent 
The implementation of the cutset data structure is discussed in Section \ref{s:cutset}.

\section{Fully dynamic connectivity algorithm}
\begin{theorem}\label{t:main}
For any constant $c$ and any sequence of polynomially many updates, there is a dynamic connectivity algorithm which answers all queries correctly with high probability and uses $O(\log^4 n)$ time per update and $O(n \log^2 n)$ words. The error is 2-sided, i.e., "yes" may be output when the correct answer is "no" and vice versa, with probability less than $1/n^c$.

If a list of current edges in the graph is maintained (with an additional $O(m)$ words), 
any sequence of any length of updates may be processed so that with high probability any single query is correct.  In this case, the error is 1-sided; if an answer to a query is ``yes" the answer is always correct and if it is ``no", it may be incorrect with probability $1/n^c$. 
\end{theorem}

\subsection{Tiers and Invariants}
 
Recall Bor\r{u}vka's parallel algorithm for building a minimum spanning tree \cite{boruvka}. We may view that algorithm as constructing a sequence of forests  we call {\em tiers}. On {\em tier} 0 are the set of nodes in the graph; on {\em tier} $\ell+1$ is a subforest of the minimum spanning forest formed when each $\ell-1$ tree picks a minimum cost edge linking it to another $\ell-1$ tree, 
\old{while edges which form cycles are discarded. If we represent each tree on each tier $\ell$ as a node, and make a node a parent of the trees from tier $\ell-1$ which form it, we have what has been called the {\em Bor\r{u}vka tree} (or forest if the graph is not connected) \cite{king}. }

Our algorithm has similar structure.   
 For each tier $\ell=0,1,...,top$, where $top=O(\log n)$,  we maintain a cutset data structure $\mathcal{CD}_{\ell}$ for $G$ with a forest $F_{\ell}$ and refer to the collection as $\mathcal{CD}$. Each cutset data structure in $\mathcal{CD}$ is generated using
 independent randomness.
Because tier $top$ is not used to form a forest on
a higher tier, the data structure only consists of a forest $F_{top}$, i.e., there is no
randomness array on this tier. 

For each tier $\ell< top$,  $\Search(T, \ell)$  refers to the $\Search(T) $ operation in the cutset data structure $\mathcal{CD}_{\ell}$ where $T \in F_{\ell}$; we similarly extend the definitions of other operations of a cutset data structure for each tier $\ell$.  Since
there is no randomness array on tier $top$, for $\Deleteedge(e,top)$ we remove $e$ from $F_{top}$ if it appears there but do nothing else.

As the graph is updated, the cutset data structure for  tier $\ell$ is used to find edges leaving non-maximal trees (``fragments") in that tier. 
Given an edge $e$, let ${\ell}$ be the minimum tier in which $e$ first appears in $ F_{\ell}$. We then refer to it as a {\em tier $\ell$ edge.} 
We maintain the following invariants of $\mathcal{CD}$. Let $G$ be the current graph.  
 

\begin{enumerate}
\setlength\itemsep{0em}
\item $F_0=V$, the vertices of $G$; 
\item For every tier $\ell$, $0\leq \ell \leq top-1$, $F_{\ell} \subseteq F_{\ell+1} \subseteq G.$
\item Let $T$ be a node in tier $\ell<top$. If $\Search(T,\ell)$ is successful  then $v(T)$ is properly contained in some $v(T')$ for $T'$ a node  on tier $\ell+1$. \label{InvariantProper}  \label{invariance3}

\end{enumerate}

\noindent
{\it Discussion:} Note that these invariants do not necessarily preserve the property that every edge in the spanning tree $F_{top}$ is currently an edge
leaving some fragment in some $\mathcal{CD}$. Rather, Invariant \ref{InvariantProper} ensures that any fragment $T$ in tier $\ell$ for which $\Search(T)$ is successful has {\it some} edge leaving it in $F_{\ell +1}$. Also, the invariants do not preclude there being a tree edge in $F_{\ell+1}$ which is incident to a fragment $T$ in $F_{\ell}$ when $\Search(T)$ is not successful. The edges in the forests may be thought of as  ``historical" $\Search$ edges, edges which were once returned by $\Search$. This gives us the flexibility to perform updates quickly.

\old{
These invariants and Lemma \ref{l:singletier} imply:

\begin{lemma} \label{l:alltiersStatic}
For any graph $G$ and any constant $c$,
let $a=\lceil{\log_{4/(4-p)} n} \rceil$ and $top=\max\{2a/\alpha, 8c\ln n/\alpha \}$, where $\alpha=(1-p)/(1-p/2)$. If the invariants hold then with probability $1-1/n^c$, $F_{top}$ is a spanning forest of $G$. If the augmented cutset data structure is used, then it is always the case that  $F_{top}$ is a subtree of a spanning forest of $G$. 
\end{lemma}
}
 Markov's inequality implies that with constant probability, a constant fraction of executions of $\mathsf{Search}$ on fragments are successful: 
\begin{lemma} \label{l:singletier}
Let $p$ be defined as in Lemma $\ref{l:cutset}$. In a cutset data structure with $f_{num}$ fragments, with probability at least $\alpha=(1-p)/(1-p/2)$ the number of fragments $T$ such that $\mathsf{Search}(T)$ succeeds is at least $(p/2)f_{num}$. 
\end{lemma}

\begin{proof}
If there are $f_{num}$ fragments in the cutset data structure, then
the expected number of fragments for which $search$ fails to succeed is no greater than
$\mu=(1-p)f$ since each has probability at least $p$ of finding a tree edge. By Markov's Inequality, the
probability that $(1-p/2)f_{num}$ fail to find a tree edge is no more than $\mu/(1-p/2)f_{num}=(1-p)/(1-p/2)$.
\end{proof}

\smallskip

As any edge can leave at most two node-disjoint fragments, it is easy to see the following:
\begin{observation} \label{observation}
Let $F$ be a forest containing $f$ fragments and $E'$ be a set of edges leaving at least $cf$ fragments $T$ then the spanning forest of $F\cup E'$ has no more than $(1-c/2)f$ fragments.
\end{observation}
\begin{lemma} \label{l:alltiersStatic}
For any graph $G$ and any constant $c$,
let $a=\lceil{\log_{4/(4-p)} n} \rceil$ and $top=\max\{2a/\alpha, 8c\ln n/\alpha \}$, where $\alpha=(1-p)/(1-p/2)$. If the invariants hold then with probability $1-1/n^c$, $F_{top}$ is a spanning forest of $G$. 
\end{lemma}

\begin{proof}
 Let $p$ be the constant in Lemma \ref{l:cutset}.
We call a tier with $f_{num}$ fragments {\em successful} if the number of fragments $T$ where $\Search(T)$ is successful is
at least $ (p/2)f_{num}$. Let $X_\ell=1$ if tier $\ell$ is successful and 0 otherwise. Then by Lemma \ref{l:singletier}, $Pr(X_\ell=1) \geq \alpha$.  For any tier $\ell$, if there are $(p/2)f_{num}$ fragments with successful searches in a tier $\ell$ then by Invariant \ref{InvariantProper}, each of these fragments is properly contained in $F_{\ell+1}$, i.e., the set of tier $\ell+1$ edges includes an edge leaving for each of these. Then by Observation~\ref{observation}, the number of fragments in $F_{\ell+1} \leq (1-p/4)f_{num}$.

Since $(1-p/4)^a n \leq 1$, $a$ successful tiers suffice to bring the number of fragments to 0 (as there can't be only 1). If we flip $\max\{2a/\alpha, 8c\ln n/\alpha \}$ coins with probability of heads $\alpha$, the expected number of heads at  $\mu=\max\{2a, 8c \ln n\}$.  The probability that the number of heads is less than $a=\mu/2$ is 
is less than
$e^{-(1/2)^2 \mu/2}=e^{-((1/2)^2 (8 c\ln n) /2 )}=  1/n^{c}$, by a Chernoff bound. 
\end{proof}

We note that the results for a fixed graph can be extended to hold for $poly(n)$ arbitrary subgraphs, including those generated by a sequence of polynomially many updates:


\begin{corollary} \label{c:dynST} Let $G^u$ be the graph $G$ after $u$ updates,
let $c',d$ be any constants, and assume $u \le n^d$. If $top=\max\{2a/\alpha, 8(c'+d)\ln n/\alpha \}$,
then
in the cutset data structures $\mathcal{CD}^u$ for each $G^u$,  with probability $1-1/n^{c'}$  every $F^u_{top}$ is a spanning forest of $G^u$. This holds
even if there is dependence in the randomness between any pair of $\mathcal{CD}^u$'s. 
\end{corollary}

\begin{proof}
Consider a collection of $n^d$ graphs $G^i$.
Set $c=c'+d$
in Lemma \ref{l:alltiersStatic}. By a union bound, 
the probability that some $F^u_{top}$ fail to be a spanning forest of $G^i$ is no greater than
$n^d (1/n^c)=1/n^{c'}$. 
 \end{proof}

\subsection{Performing updates}

We say a fragment in a tier $\ell$ is {\it isolated} if it is equal to 
a fragment in tier $\ell+1$. 
To efficiently maintain the invariants as edges are updated,  we only check a cutset induced by a tree if the invariants could have become violated. The invariants are maintained from the bottom up. Changes to a tree may be needed because (1) a fragment becomes isolated when a tree edge is deleted; (2) for an isolated fragment $T$, $\Search(T)$ may become successful, when before it wasn't or (3) a fragment which is newly formed
is isolated.

Let $T(x)$ denote the tree containing node $x$ in a forest $F$. If edge $\{x,y\}$ is inserted, we start at the lowest tier, checking on each tier $\ell$ as follows:  if  $T(x) \neq T(y)$ this may affect $\Search(T(x),\ell)$ and $\Search(T(y),\ell)$. If $T(x)$ is isolated and 
$\Search(T(x),\ell) $ is successful, the edge returned by $T(x)$ is added to $F_{\ell+1}$. Once a tree edge is added it must be added in all higher tiers.
If this forms a cycle, a highest tier edge in the cycle is made into a non-tree edge. This may in turn cause an isolated fragment $T'$ on a higher tier $\ell'$, for which $\Search(T',\ell')$  must be checked. 
If edge $\{x,y\}$ is deleted, it does not affect the results of a search if both $x$ and $y$ are in the same fragment. If they are not,
then calls to $\Search$ for both $x$ and $y$ are potentially affected. At this point we
proceed as we do in the case of insertions.


\medskip

We now give the algorithms for handling updates and queries. The following procedure, which restores the invariants after an update involving an edge $\{x,y\}$, as described
above, is used for both insertions and deletions.

\begin{algorithm}[h!] 
\caption{$\Refresh(\{x,y\})$}
\label{a:restore}
\begin{algorithmic}[1]
\FOR{$\ell=0,\dots,top-1$}
\FOR{$u \in \{x,y\}$}
\IF{$T(u)$ is isolated and $\Search(T(u),\ell)$ returns an edge $e_u=\{a,b\}$}
\IF{a path between $a$ and $b$ exists on some tier $> \ell$} 
\STATE{Let $j$ be the lowest tier on which there is a path}
\STATE{Let $e'$ be an edge of maximum tier in the path on tier $j$}
\STATE{Call $\Insertedge(e',k)$ for each $k \ge j$}
\STATE{Call $\mathsf{Make\_tree\_edge}(e_u,\ell')$ for each  $\ell'$,
$top \ge \ell' > \ell$} \label{lineTree}
\ENDIF
\ENDIF
\ENDFOR
\ENDFOR
\end{algorithmic}
\end{algorithm}

\begin{itemize}
\setlength\itemsep{0em}
\item {\sf Initialize}:
For each tier $\ell=0,1,2,...,top-1$, initialize a cutset data structure $\mathcal{CD}_{\ell}$   where $F_{\ell}$ contains $n$ trees consisting of the single vertices in $G$.

\item
$\mathsf{Insert}(e)$:
For $\ell < top$, call $\Insertedge(e,\ell)$ to insert $e$ into $\mathcal{CD}_{\ell}$, 
and call $\Refresh(e)$.

\item
$\mathsf{Delete}(e)$:
\label{a:deletions}
For $\ell \le top$, call $\Deleteedge(e,\ell)$
to
delete $e$ from $\mathcal{CD}_{\ell}$, 
 and call $\Refresh(e)$.

\item
$\mathsf{Query}(x,y)$:
If $T(x)=T(y)$ in $F_{top}$, return ``yes", else return ``no". 

\end{itemize}
A proof of the following is given in the Appendix.
\begin{lemma} \label{l:invariants}
\label{correct}
Invariants (1), (2) and (3) are maintained by $\mathsf{Initialize}$, $\mathsf{Insert}$, and $\mathsf{Delete}$.
\end{lemma}

\subsection{Implementation details, running time and space analysis}
We briefly review basic data structures used.
\label{ss:ET}
\old{We use ET-trees from \cite{HK} to store each tree $T \in F$.} 
ET-trees are based on the observation that when a tree is linked with another tree, the sequence given by the Euler tour of the new tree can be obtained by a constant number of links and cuts to the sequences of the old trees. The sequence is stored in a balanced ordered binary tree.

\medskip

 \noindent
 {\it ET-tree properties:} If each node in the graph has a value stored in $x$ words of length $O(\log n)$ then the following hold.
 \begin{enumerate}
\setlength\itemsep{0em}
 \item
A tree edge can be deleted or inserted in $O(x\log n)$ time.

\item
The sum of the values stored in the nodes of a given tree can be returned in $O(x)$ time.
  
\item
 A value at a node can be updated in $O(x\log n)$ time per word. 
   
   \item
  Given a node $v$, the name of the tree $T(v)$ containing
$v$ can be returned in $O(\log n)$ time. 

\item
 The name of the tree can be obtained more quickly if one keeps an ET-tree of degree $O(\log n)$ solely for this purpose. That is, it can be obtained in $O(\log n/\log \log n)$ if one is willing to spend $O(\log^2 n)$ time doing an insertion or deletion of a tree edge. 
 
 \end{enumerate}

{\it ST-trees} (aka link-cut trees) provide a means of finding the maximum weight edge on a path between two nodes in a dynamic forest, in $O(\log n)$ time per tree edge link, cut, and the find operation. 
 ST-trees and ET-trees require only $O(n*|value| )$ space where $|value|$ is the size of the vector or value stored at each node or edge.

\bigskip

\noindent
{\it Implementation details:} $F_{top}$ is stored as a degree-$\lg n$ ET-tree where no values are associated with a node, so that queries can be answered in time $O(\log n/\log \log n)$.
$F_{top}$ is also stored  in an ST-tree where each tier $\ell$ tree edges has weight $\ell$. An ST tree data structure allows us to find the cycle edge in $O(\log n)$ time. 

In the worst case, a new tree edge is found in every tier twice and these changes are filtered up, for a total of
$2 \sum_{\ell} top-\ell$  tree edge insertions or deletions:

\begin{lemma} \label{l:analysis}

Let $t(n)$ be the update time for a cutset data structure and $s(n)$ be the space. Then the running time and space for
the fully dynamic connectivity algorithm is $O((top^2 *t(n)+ top*log n)$ and the space needed is $O(top*s(n)+ n)$. 
\end{lemma} 

Since $top=O(\log n)$, Lemma \ref{l:analysis} and Lemma \ref{l:cutset} imply Theorem \ref{t:main}.

\section{Implementation of the Cutset Data Structure }\label{s:cutset}
To illustrate the main idea of our data structure, we will first explain a simplified deterministic version which works when there is exactly one edge in each cutset. 

Let $G=(V,E)$ be the current graph. 
Each edge $\{x,y\} \in G$, $0 \le x<y < n$, has the {\it  name} $x_b \cdot y_b$ where $x_b$ and $y_b$ are the $\lceil \lg n \rceil$
bit vectors containing the binary representations of nodes $x$ and $y$ respectively and $x_b \cdot y_b$ denotes the $2 \lceil \lg n \rceil$ bit vector $x_b$ followed by $y_b$. We say that edge $\{x,y\}$ is {\it given by} $x_b \cdot y_b$.  Each node will maintain a vector equal to the bitwise
XOR (i.e., sum mod 2) of the vectors assigned to each of its incident edges
and we use summation notation to denote bitwise XOR of two or more vectors.
For each tree $T \in F$, the ET-tree  for $T$ maintains the bitwise XOR of the vectors stored at $T$'s nodes.
We use ``addition" or ``subtraction" of vectors to refer to the bitwise XOR of the vectors.
The operations are implemented as follows:

\begin{itemize}
\setlength\itemsep{0em}
\item 
$\mathsf{Insert\_edge}(x,y )$: Add its name $x_b \cdot y_b$ to the vectors at nodes $x$ and $y$.

\item
$\mathsf{Make\_tree\_edge}(x,y)$: Insert $\{x,y\}$ into $F$ (by linking the ET-trees representing $T(x)$ and $T(y)$).

\item $\mathsf{Make\_nontree\_ edge}(x,y)$: Delete $\{x,y\}$ from $F$ by removing it from the ET-tree that contains it. 

\item
$\mathsf{Delete\_ edge}(x,y)$: Subtract vector $x_b \cdot y_b$ from the vectors at the nodes $x$ and $y$. If the edge is in $F$, remove it from $F$ by calling $\mathsf{Make\_nontree\_edge}(x,y)$.

\item
$\mathsf{Search}(T)$: . Let ${z}=z_1z_2...z_{ 2\lceil \lg n \rceil}$ be the sum of the vectors of nodes  in $T$.  If ${z}\neq \vec{0}$ then return ${z}$. 

\end{itemize}
\begin{lemma}
\label{basic-observation}
 For any tree $T \in F$, if there is exactly one edge in the cutset $C$ of the cut $(T, V\setminus T)$, then $\Search(T)$ is always successful.  The amount of space used is $O(n\log n)$ bits. 
\end{lemma}

\begin{proof}
Every edge with two endpoints in a tree $T$ is added to the vectors of two nodes in $T$ and contributes twice to  ${z}$, thus it appears not all in the sum ${z}$, while every edge with exactly one endpoint contributes exactly once to ${z}$. If the cutset has exactly one edge, ${z}$ is the name of the edge. 
\end{proof}

\subsection{Extension to cutsets with more than one edge} 
We extend our data structure to handle the case of when  the cutset of $(T,V\setminus T)$ has more than one edge. 
We may do this in sublinear space, but then we require the assumption that no edge is deleted unless it is currently in the graph and the number of updates is bounded by $n^d$ for some constant $d$. Alternatively, we can keep a list of edges currently in the graph.

We maintain for each node $x \in V $ an array indexed by {\it levels} $i$ where $0\leq i < \lceil{2\lg n}\rceil=levelNum$.  
Let $s_{i}(x)$ denote the vector on level $i$ for node $x$.
When edge $\{x,y \}$ is inserted into $G$, for  each $i = 0,1,...,levelNum-1$ with probability $1/2^i$ we add $x_b \cdot y_b$ to  both $s_{i}(x)$ and $s_{i}(y)$. 
The idea is that when the size of the cutset is $\approx 2^i$,  with constant probability there is exactly one edge from the cutset and a unique $x$ in $T$, such that the name of the edge is added to $s_{i}(x)$.  In this case, $z_i=\sum_{x \in T} s_{i}(x)$ is the name of that one edge and $\Search(T)$ is successful. 

To decide which edges to included in the sums on the various levels without recording each coinflip, we use a 2-wise independent hash function $h$ mapping
edges into $[2^{levelNum}]$, where for $k \ge 1$, $[k]$ denotes
$\{1,\dots,k\}$.
An edge $e$ is sampled in level $i$, $i=0,...,levelNum$ of the cutset data structure iff  $h(e) \leq 2^i$.
 

The following lemma is straightforward and has been observed previously. A proof is in the Appendix.
\begin{lemma}  \label{l:findany} Let $W$ be the cutset of $(T, V \setminus T)$. With probability $1/8$,
there is an integer $j$
such that exactly one $e\in W$ hashes to a value in $[2^j]$.
\end{lemma}

\old{If we keep a list of current edges in the graph (using space linear in the number of edges) then we can easily check deterministically if $z_i$ the name of an edge and if the edge is in a cut. Alternatively, we can save space.}
To verify with high probability that $z_i$ is actually the name of one edge  in the cut rather than the sum of several names, 
we make use of an {\it odd hash function} \footnote{ We may view this verification problem as an instance of {\em 1-sparse recovery} and it 
is possible to apply known techniques for this problem in our setting. In particular \cite{CormodeM} proposes an approach which uses polynomial evaluation modulo
a prime $p$ to verify uniqueness for recovered values (from a set of size $t$,) with success probability $(1-t/p)$. This method maintains  sums of the form $\sum_{i\in S}a_iz^i$, where $S \subseteq [t]$,  $a_i $ is the {\em weight} of item $i$, and $z \in \mathbb{Z}_p$ is  randomly chosen.
Using this approach for our purposes, in order for $\Search$ to maintain correct values with probability $(1-1/n^c)$ over a sequence of $n^d$ edge updates, we must work over $p=\omega(n^{c+d+2})$. Adding an edge would require adding (mod $p$) a term $\pm z^i$, where $1 \le i \le \binom{n}{2}$. To the best of our knowledge,
this requires performing roughly $2\lg n$ multiplications of $(c+d+2)\lg n$-bit numbers, with
reduction modulo $p$, or alternatively, $c+d$ independent parallel repetitions of the technique with $(c+d)$ primes of size $> 4 \lg n$ bits.  In contrast, our approach may be done with $(c+d)\lg n$ multiplications of $2 \lg n$-bit random numbers, with reduction modulo a power of $2$.}

We say that a random hash function $f:[1,m]\rightarrow
\{0,1\}$ is {\em $\epsilon$-odd}, if for any given non-empty set
$S\subseteq [1,m]$, there are an odd number of elements in $S$ which
hash to 1 with probability $\epsilon$, that is,
\begin{equation}\label{eq:non-zero}
\Pr_f\left[\sum_{x\in S} f(x)=1\mod 2 \right]\geq\epsilon.
\end{equation}

We use the construction of 
 \cite{thorup-sample} to create a $(1/8)$-odd hash function. Let $m\leq 2^w$, where $w=2\lceil\lg n\rceil$. We pick uniformly at random
an odd multiplier  $k$ from $[1,2^w]$ and a threshold $t\in [1,2^w]$.
From these two components, we define $f:[1,2^w]\rightarrow \{0,1\}$ as follows:
\begin{eqnarray*}
f(x)& = & 1 \hbox{ if  }(kx\bmod 2^w)\leq t \hbox{ and } x \neq 0 \\
       & = & 0  \hbox{ otherwise.}
\end{eqnarray*}

The idea of the verification is as follows:  we use a 2-wise independent function to randomly partition a set into two and test each part using using an odd hash function to determine if both parts have at least one element. Repeating this 
$O(\log n)$, if all tests fails, then with high probability there cannot be more than one element in a set. 

In particular, letting $S = [\binom{n}{2}]$,
the algorithm is initialized by choosing the following functions independently at random.
\begin{itemize}
\setlength\itemsep{0em}
\item
 $h_j: S \rightarrow \{0,1\}$,  $j=1,2,...,c\lg n$. Each $h_j$ is a 2-wise independent hash function. 

\item
$f_{j,b}: S \rightarrow \{0,1\}$ for $j=1,2,..., c\lg n$, and $b=0,1$. Each $f_{j,b}$ is an independently random $(1/8)$-odd hash function. 
 \old{
 $d_{j,a}: S \rightarrow S$   for $j=1,2,..., c\lg n$ and $a=0,1$,   k$ is a  random odd number between $1$ and $2^\lceil{n \choose 2} \rceil$. 
  \item
 $t_{ja}$: for $j=1,2,..., c\lg n$ and $a=0,1$, a random element of $S$.
}
\end{itemize}
The $f_{j,b}$'s are used to determine a $2c\lg n$-bit {\em tag} of
an edge at a particular level.
For each node $x$, we maintain an auxilliary vector
$s'_{i,j,b}(x)$.  The tag of an edge is added (by bitwise exclusive-or) to the vector $s'_{i,j,b}(x)$ when the edge name is added to $s_i(x)$. For each level $i$,  the sums over the nodes in a tree of the auxillary vectors are maintained as well as the sums of the vectors. Let $z'_{i,j,b} =\sum_{x \in T} s'_{i,j,b}(x)$.  The specifics are described below. 

\old{
If for any f both bits in a pair equal 1, this serves as proof that $z_i$ is the sum of more than one name, while if no pair has this property, then with high probability $z_i$ is the name of one edge. 
 The tag of an edge is added (by bitwise exclusive-or) to the vector $s'_i(x)$ when the edge name is added to $s_i(x)$. The sums over the nodes in a tree of the auxillary vectors are maintained as well as the sums of the vectors. We let $z'_i =\sum_{x \in T} s'_{i}(x)$.  The auxillary vectors are simple and fast to construct. 
 }
 \old{
 Alternatively, if the current list of edges in the graph is maintained in a binary search tree, then there is no need to use a tag to verify that $z_i$ is an edge in the cut. One can deterministically find out if $z_i$ is in the list of current edges and use the ET-tree representation of $F$ to test that one of its endpoints is in $T$ and another lies outside $T$.  }

The operations $\mathsf{Make\_tree\_edge}(x,y)$ and 
$\mathsf{Make\_nontree\_ edge}(x,y)$ are implemented as in the case where
cutsets have size one. The remaining operations are implemented as follows: 

\begin{itemize}
\setlength\itemsep{0em}
\item
\begin{tabbing}
$\mathsf{Insert\_edge}(x,y)$:\\
For  $i=0,1,2,\ldots, levelNum-1$\\
\ \ \= If $h(e) \leq 2^i$ then \\
\> \ \  --\{{\it Add $e$ to sample} \} add its name $x_b \cdot y_b$ to value $s_{i}(x)$ and $s_i(y)$ and \\
\> \ \ \= For $j=1,...,c'\lg n$ and $b=0,1$:\\
\> \> \ \  --\{{\it Add tag for $e$:}\}  If $h_j(e)=b$ and $f_{j,b}(e)=1$ then add 1 to $s'_{i,j,b}(x)$ and $s'_{i,j,b}(y)$. 
\end{tabbing}

%
%

\item
\begin{tabbing}
$\mathsf{Delete\_ edge}(x,y)$:\\
For  $i=0,1,2,.., levelNum -1 $:\\ 
\ \ \= If $h(e) \leq 2^i$ then \\
\> \ \ --subtract its name $x_b \cdot y_b$ to value $s_{i}(x)$ and $s_i(y)$ and \\
\> \ \ \=  --For  $j=1,...,c'\lg n$ and $b=0,1$:\\
\> \> \ \ \ \ If $h_j(e)=b$ and $f_{j,b}(e)=1$ then subtract 1 from $s'_{i,j,b}(x)$ and $s'_{i,j,b}(y)$.\\
If the edge is in $F$, call $\mathsf{Make\_nontree\_edge}(x,y)$ to remove it from the ET-tree containing it. \\
\end{tabbing}
\item
\begin{tabbing}
$\mathsf{Search}(T)$:\\
For $i=0, 1, 2, ..., levelNum-1$:\\ 
\ \ \= While an edge in the cutset $(T, V\setminus T)$ has not been found, test $z_{i}$ 
as follows:\\
\> \ \ \= Let $\{z^1,z^2\}$ be the edge given by the minimum $i$ such that  $z_{i} \neq 0$. \\
\> \>  If there is no pair $j$ such that $z'_{i,j,0}=z'_{i,j,1}=1$, then it is an edge in the cutset and returned. \\
\> \> Else $null$ is returned.  
\end{tabbing}

\end{itemize}

Let ${z_i} $ be the sum of the level $i$ vectors for a tree $T$, and $i$ be the minimum level such that ${z_i} \neq  \bar{0}$. Then with constant probability $p$, ${z_i}$ is the name of an edge in the cut $(T, V\setminus T)$. If ${z_i} \neq 0$ then if 
if ${z_i}$ is not the name of an edge in $G$, then with high probability, for some pair $j$, $z'_{i,j,0}=z'_{i,j,1}=1$. 

If ${z_i} \neq 0$ there is must be at least one edge in the cut, since otherwise $h(e)=\bar{0}$ for all edges $e$. 
If there is more than one edge in the cut, let $e$ and $e'$ be two such edges. For $b=0,1$, let  $S_b=\{e ~|~h_j(e)=b\}$. Since each $h_j$ is 2-wise independent then for any $j$, $h_j(e) \neq h_j(e')$ with probability 1/2, and therefore with probability 1/2, $|S_0|> 0$ and $ |S_1| >0$. 
Recall that $h_j, f_{j,0}$ and $f_{j,1}$ are independently chosen random hash functions. Therefore, given that $|S_0| >0$ and $|S_1| >0$,  
$z'_{i,j,0}=z'_{i,j,1}=1$ with probability $(1/2)(1/8)^2=1/128$.
The probability that this equality fails for all $j=1,2,..., c(128) \lg n$ pairs is $(1-1/128)^{c (128 \lg n)}$ which is less than $1/n^c$ for sufficiently large $n$.  Therefore w.h.p., when $\Search$ returns an edge it is an edge in the cut. 
The probability of correctness as stated in Lemma \ref{l:cutset} then follows. 

\old{

\begin{theorem} \label{t:main}
For any polynomial sequence of updates and queries, the low space data structure answers all queries correctly with high probability and uses $O(\log^4 n)$ time per operation and $O(n \log^3 n)$ bits of space. 
\end{theorem}

\begin{proof}
The correctness follows from Lemma \ref{} since each edge is found with a somewhat worse but constant probability $p"$.
In addition, a $1/n^{c'}$ probability error is introduced with each verification, but over polynomial updates, for sufficiently large constant $c'$,  this is still less than $1/n^c$.

The time to perform an update is affected only by a constant factor as $O(\log n)$ bits are stored per vertex per level in a cutset data structure as before. The cost of checking the verification requires an additive $O(\log n)$ time for each $search$, which does not affect the asymptotic running time. 
\end{proof}
}

\bigskip

\noindent
{\it Running time for the cutset data structure}

Each hash function requires a constant time to perform the hash. $O(\log n)$ hashes are performed to determine if an edge should be sampled for level $i$ and for computing its tag. 
Each node in a tree in the forest of a cutset data structure has a vector of $O(\log n)$ words, since there are $O(\log n)$ levels for each level there are two $O(\log n)$ bit vectors, $s_i(x) $ and $s'_i(x)$.  Maintaining the sum of vectors of in an ET-tree of $O(\log n)$ vectors of size $O(\log n)$ requires $O(\log^2 n)$ time. Similarly, $\mathsf{Make\_nontree\_edge}$ or $\mathsf{Make\_tree\_edge}$ requires $O(\log^2 n)$ time. 
$Search(T)$ involves finding the maximum level $i$ with non-zero $z_i$ and checking its tag  or determining if $z_i$ is on the current edge list  and determining the name of the trees containing these endpoints all take no more than $O(\log n)$ time each.
Hence we can conclude that the update time for edge insertion or deletion is $O(\log^2 n)$ time while the search time is $O(\log n)$.
 This concludes the proof of Lemma \ref{l:cutset}.

\section{2-edge-connectivity}

In this section, we treat the dynamic 2-edge connectivity algorithm \cite{Holm} and two instances of the space efficient dynamic connectivity data structure presented previously as black boxes, which we denote by 2-EDGE, CONN1, and CONN2 respectively. All are generated using independent randomness. Using these black boxes, we will present an algorithm to solve dynamic 2-edge connectivity in $O(n \log^2 n)$ words with the same query time and an update time which is  within a $O(log^2 n)$ factor of \cite{Holm}, i.e.,
 $O(\log^6 n)$ amortized update time and $O(\log n/\log\log n)$ query time.  The idea is to maintain a certificate for 2-edge connectivity using the dynamic connectivity method shown in the previous sections, and use that certificate as the input into 2-EDGE. In the algorithm below, 
We give some definitions:

\begin{itemize}
\setlength\itemsep{0em}
\item
$G=(V,E)$: the dynamic graph
\item
$F_{1}$: A maximal forest of G maintained by CONN1
\item
$G_{1}=(V, E_{1})$: The graph with the forest edges removed. $E_{1} = E \setminus F_1$
\item
$F_{2}$: A maximal forest of $G_{1}$ maintained by CONN2
\item
$G_c=(V, E_c)$ The graph which is the input to 2-EDGE.  
\item
$F$: The 2-edge connected forest  of $G_c$ maintained by 2-EDGE
\end{itemize}

We now present algorithms for handling updates and queries:

\noindent
{\it Insert(e)}:\\
1. Insert $e$ into $E$ and $E_1$\\
\{Process updates to $F_{1}$\}\\
2. For each edge $e'$ added to $F_1$, 
delete $e'$ from $E_{1}$\\
3. For each edge $e'$ removed from $F_1$
insert $e'$ into $E_{1}$\\
\{Process updates to $F_1 \cup F_{2}$\}\\
4. For each edge $e'$ added to $F_1 \cup F_2$, insert $e'$ into $E_{c}$\\
5. For each edge $e'$ removed from $F_1 \cup F_2$, remove $e'$ from  $E_{c}$\\

\noindent
{\it Delete edge $e=(u,v)$} is the same as $Insert(e)$ except for the first step which is replaced by
``Delete $e$ from $E$ and $E_1$".\\

\noindent
$Query(u,v)$: Query 2EDGE to see if $u$ and $v$ are 2-edge connected.

\begin{theorem}
\label{2edge}
W.h.p., fully dynamic 2-edge connectivity can be maintained in $O(\log^6 n)$ amortized time per update and $O(\log n/\log \log n)$ per query using space $O(n \log^2 n)$ words, over a polynomial length sequence of updates.
\end{theorem}

\noindent
A proof is given in the Appendix.

\old{
 \section{Discussion}
The methods here are surprisingly simple yet quite different from other known techniques for dynamic graph algorithms. As noted in the related work section, this work can be seen as a technique for adaptive graph sketching with a concern for fast update and query time. 
What update and query times can be achieved for other graph problems in the graph sketching model or when space is an issue? It is also interesting to ask about these questions in a parallel streaming model such as MapReduce.

Classic dynamic graph problems which remain open and which still seem hard are dynamic minimum spanning tree in 
$o(\sqrt{n})$ worst case time and Las Vegas or deterministic graph connectivity with $o(\sqrt{n})$ worst case cost. 

 It is not hard to see that the technique described here can be made deterministic with an additional $\tilde{O}(k)$ factor in the update time  if we know the cuts are of size no greater than $k$, through the use of combinatorial designs. Hence any lower bounds for deterministic algorithms should make use of large cuts.\\}

 \noindent
 

\newpage

 \bibliographystyle{plain}


\section*{Appendix}
\section*{Additional proofs}
\newtheorem*{invariants}{Lemma 3.5}
\begin{invariants} 
Invariants (1), (2) and (3) are maintained by $\mathsf{Initialize}$, $\mathsf{Insert}$, and $\mathsf{Delete}$.
\end{invariants}

\begin{proof}
It is clear that all invariants hold after a call to $\mathsf{Initialize}$. We now consider update operations.

Maintenance of Invariant (1) holds trivially because $\Refresh$ never makes an edge
into a tree edge on tier 0.

For Invariant (2), first note that $\Refresh$ checks for a cycle and removes an edge from a potential cycle before inserting an edge which causes a cycle in a tree, and so each $F_{\ell}$ is a forest. For the inclusion
property $F_{\ell} \subseteq F_{\ell+1}$,  first note that when handling deletions, the same edge is
deleted from every tier, so this does not violate the invariant. For both
insertions and deletions,
 line \ref{lineTree} of $\Refresh$ makes $e_u$ a tree edge on all tiers
$> \ell$, so inclusion is immediately maintained when no cycles are created. 
Now suppose edge $e_u$ would form a cycle if inserted.
 Let $e'$ be an edge of maximal tier $j$ in this cycle. 
We first note that $j > \ell +1$, due to the fact that $T(u)$ is isolated. So by
the minimality of $j$ and the maximality of the tier of $e'$, 
it must be the case that $e'$ is not an edge in the forest on any tier $< j$. Therefore, $e'$ is removed from all forests containing it, and the
 inclusion property of Invariant 2 is maintained.

Invariant (3) fails only if there is an isolated fragment $T$ on some tier $\ell$ for which $\Search(T,\ell)$ is successful. We call this a {\em bad} $T$.  Initially, if we start with an empty graph, there is no bad $T$.  We assume there is a first update during which a bad $T$ is created and show a contradiction. A bad $T$ is created either if (1) $T$ is newly formed by the update or (2) $T$ pre-exists the update and $\Search(T,\ell)$ is successful in the current graph $G$. In (2) there are two subcases: $T$ wasn't isolated before the update or $T$ was isolated before the  update but $\Search(T,\ell)$ became successful after the update. 
Let $\{x,y\}$ be the updated edge. 
We first show: \\

\noindent
{\it  Claim:} Before the call to $\Refresh$, the only bad fragment introduced on
any tier has the form $T(x)$ or $T(y)$,  and this remains the case during the execution of $\Refresh$.\\
\noindent
{\it Proof of Claim:}
Before the call to $\Refresh$, there are two
ways in which this update could create a bad fragment. The first, which may
occur for both insertions and
deletions, happens when there is a fragment $T$ on a tier $\ell$ for which
$\Search(T,\ell)$ did not return an edge before the update, but now does.
This can only be the case if the cut induced by $T$ has changed, and this is only when exactly one of $x$ or $y$ is a node in $T$.
The second can happen before a $\Refresh$ and only for deletions, in particular when $\{x,y\}$ is removed. In this case, two new fragments  $T(x)$ and
$T(y)$ are created and no other cuts are changed.  Hence the claim is proved for the bad fragments created before the call to $\Refresh$. 

Now, we consider what happens during $\Refresh$. 
During the execution of $\Refresh$, a bad fragment may result when $e_u$ is made into a tree edge on some tier $j > \ell$ on Line \ref{lineTree}, causing
two fragments $T$ and $T'$ to be joined into one new fragment which may be bad. Since
$e_u$ is the result of $\Search(T(x),\ell)$ or $\Search(T(y), \ell)$ for some $\ell$, the new fragment contains $x$ or $y$, and by Invariant 2, so does every tree on higher tiers created by inserting $e_u$.   Note that replacing one edge by another edge in a cycle does not affect any cut, so the replacement of $e'$ by $e_u$ does not cause the creation of a new bad fragment. This concludes the proof of the claim.
\old{ 
it will create a fragment with the same set of nodes as a fragment that existed
before the update. So the cut induced by these fragments will be the same,
and the fragment resulting after deleting the cylce edge will be bad only if
the fragment was already bad, in which case it must have the form $T(x)$ or $T(y)$.} \\

Given the claim, we now show that after calling $\Refresh$, Invariant (3)
holds. The proof is by induction on the value of $\ell$ in the for-loop of $\Refresh$. Assume $\ell <top$.  Suppose that before iteration $\ell$, Invariant 3 holds on all tiers $< \ell$ or $\ell=0$. 
This will still be the case after iteration $\ell$, since this iteration only makes changes
on tiers $> \ell$. Suppose that $T(x) \in F_{\ell} $ is bad. This means that
$T(x)$ is isolated, and $\Search(T(x),\ell)$ returns an edge $e_u$.
Since $e_u$ is inserted into $F_{\ell+1}$,  $T(x)$ in tier $\ell$ is no longer bad,
and Invariant (3) holds on all tiers $\le \ell$.
The same argument applies to $T(y)$. Maintenance of Invariant (3) now holds by induction for all $\ell <top$.
\end{proof}
Lemma~\ref{l:findany} follows immediately from the following.
\begin{lemma}  With probability $1/8$,
there is an integer $j$
such that exactly one $w\in W$ hashes to a value in $[2^j]$.
\end{lemma}
\begin{proof}
Let $\ell=\binom{n}{2}$.
We prove the statement of the lemma for $j=\ell-\lceil{\lg |W|}-2 \rceil$.
Then $1/(4|W|)<2^j/2^\ell < 1/(2|W|)$. 
\begin{align*}
\Pr_{h}&[\exists !\, w\in W: h(w)\in [2^j]]\\[-1mm]
&= \sum_{w\in W} \Pr_{h}[h(w)\in [2^j]
\wedge \forall w'\in W\setminus\{w\}: h(w')\not\in [2^j]]\\[-3mm] \\
& \hbox{We evaluate the inside of the summation: }\\
\old{
&A(w) \hbox{ denote } \Pr_{h} [h(w)\in [2^j]]:\\
&=  \sum_{w\in W} A(w) \Pr_{h}[\forall w'\in W\setminus\{w\}:h(w')\not\in [2^j] \mid h(w)\in [2^j]]\\[1mm]
&\geq \sum_{w\in W} A(w)
(1-
\sum_{w'\in W\setminus\{w\}}\Pr_{h}[h(w')\in [2^j]\mid h(w)\in [2^j]])\\
 &\hbox{By 2-wise independence:} \\
&=\sum_{w\in W} A(w)
(1-
\sum_{w'\in W\setminus\{x\}}\Pr_{h}[h(w')\in [2^j]])\\ 
&= |W|(2^j/2^\ell\cdot (1-(|W|-1)2^j/2^\ell)) \\
&>|W|/(4|W|)(1-|W|/(2|W|)=1/8.}
&=  \Pr_{h} [h(w)\in [2^j]]\Pr_{h}[\forall w'\in W\setminus\{w\}:h(w')\not\in [2^j] \mid h(w)\in [2^j]]\\[1mm]
&\geq  \Pr_{h} [h(w)\in [2^j]]
(1-
\sum_{w'\in W\setminus\{w\}}\Pr_{h}[h(w')\in [2^j]\mid h(w)\in [2^j]])\\
 &\hbox{By 2-wise independence:} \\
&=  \Pr_{h} [h(w)\in [2^j]]
(1-
\sum_{w'\in W\setminus\{w\}}\Pr_{h}[h(w')\in [2^j]])\\ 
&= (2^j/2^\ell\cdot (1-[(|W|-1)2^j/2^\ell]) \\
&>(1/(4|W|)(1-[(|W|-1)/(2|W|)] >  1/(4|W|)(1/2)]=1/(8|W|)\\
\hbox{Then} \Pr_{h}&[\exists !\, w\in W: h(w)\in [2^j]]> 1/8.
\end{align*}
\end{proof}
\newtheorem*{twoedge}{Theorem 5.1}
\begin{twoedge}
W.h.p., fully dynamic 2-edge connectivity can be maintained in $O(\log^6 n)$ amortized time per update and $O(\log n/\log \log n)$ per query using space $O(n \log^2 n)$ words, over a polynomial length sequence of updates.
\end{twoedge}
\begin{proof}
First we prove correctness:
Let $F$ be a spanning forest of $G$ and $F'$ be a spanning forest of $G\setminus F$. Then $F\cup F'$ is a certificate of
2-edge connectivity, i.e., if $G$ is 2-edge connected then so is  $F\cup F'$. The correctness then depends on
whether $F_1$ and $ F_2$ are correctly maintained as the spanning forests of $G$ and $G \setminus F_1$ respectively.
By the proof of Theorem \ref{t:main}, $F_1=F_{top}$ maintained by CONN1 is the spanning forest of $G$ w.h.p. Since CONN2 uses independent randomness from CONN1,
the updates to $F_1$ and therefore the updates of the input graph $G\setminus F_1$ to CONN2 are independent of the randomness in CONN2, similarly
$F_2=F_{top}$ maintained by CONN2 is a spanning forest of $G \setminus F_1$ w.h.p. By a union bound, both forests are maintained correctly w.h.p. We note that 2EDGE is always correct.

We observe from the proof of Theorem \ref{t:main} that a single update to $E$ results in $O(\log n)$ updates to $F_{top}$ in CONN1, which in turn result in $O(\log^2 n)$ updates to $E_C$. Each update to $E_C$ has amortized time
$O(\log^4 n)$ in 2EDGE. Hence the total time per update is $O(\log^4 n)$ worst case time in CONN1 plus $O(\log^5)$ worst case time in CONN2 plus $O(log^6)$ amortized time in 2EDGE, for a total amortized update time of $O(\log^6 n)$.
As the space used by 2EDGE is $O(m + n\log n)$ words and $|E_c|< 2n$, and the space used by CONN1 and CONN2 is
$O(n \log^2 n)$ words,  the total space used is $O(n \log^2 n)$ words.

\end{proof}

\end{document}